\documentclass{lmcs}
\usepackage{mik-makro,amsmath,mathrsfs,color}
\usepackage{xspace,pdfsync}
\usepackage[all]{xy}

	\newtheorem{claim}{Claim}

\title{Star height via games}
\author{Miko{\l}aj Boja\'nczyk}
\begin{document}
\maketitle
\begin{abstract} 
		This paper proposes a new algorithm  deciding the star height problem. As shown by Kirsten, the star height problem reduces to a  problem concerning automata with counters, called limitedness. The new contribution is a different algorithm for the limitedness problem, which reduces it to solving  a Gale-Stewart game with an  $\omega$-regular winning condition.
	\end{abstract}

The star height of a regular expression is its nesting depth of the Kleene star. 
The star height problem is to compute the star height of a regular language of finite words, i.e.~the smallest star height of a  regular expression that defines the language. For instance, the expression  $	(a^*b^*)^*$ has star height two, but its language has star height one,  because it is defined by the  expression $(a+b)^*$ of star height one, and it cannot be defined by a star-free expression. Here we consider regular expressions without complementation, in which case star-free regular expressions define only finite languages.

The star height problem is notorious for being one of the most technically challenging problems in automata theory. The problem was first stated by Eggan in 1963 and remained open for 25 years until it was solved by  Hashiguchi in~\cite{DBLP:journals/iandc/Hashiguchi88}. Hashiguchi's solution  is very difficult. Much work was  devoted to simplifying the proof, resulting in important new ideas like the tropical semiring~\cite{Simon88} or Simon's Factorisation Forest Theorem~\cite{Simon88b}.   This work culminated in a  simplified proof, with an elementary complexity bound, which was given by Kirsten in~\cite{Kirsten05}. Kirsten introduces an automaton model, called a {nested distance desert automaton}, and shows that the star height problem can be reduced to a decision problem for this automaton model, called the {limitedness problem}. The reduction of the star height problem to the limitedness problem is not difficult, and most of the effort in~\cite{Kirsten05} goes into deciding the limitedness problem.   Later solutions to the limitedness problem can be found in the work of Colcombet on cost functions~\cite{Colcombet09}, or  in the PhD thesis of Toru\'nczyk~\cite{Torunczyk11}.  All  are based on variants of the limitedness problem.

The contribution of this paper is a conceptually simple algorithm deciding the limitedness problem, and therefore also the star height problem. The algorithm is a rather simple reduction of  the limitedness problem to the  Church Synthesis Problem, i.e.~the problem of  finding the winner in a Gale-Stewart game with an $\omega$-regular winning condition. Since the Church Synthesis Problem itself is a hard problem, the contribution of this paper should mainly be seen as  establishing a connection between the two questions (limitedness, and the Church Synthesis Problem), as opposed to giving a simple and self-contained solution to the star height problem.

We show how the star height problem is reduced to the limitedness problem in Section~\ref{sec:reduction}, and in Section~\ref{sec:limitedness} we present the new contribution, namely the algorithm for limitedness.  
 
 I would like to thank Nathana\"el Fijalkow, Denis Kuperberg, Christof L\"oding and the anonymous referees for their helpful comments.  This paper is a journal version of a paper presented at LICS 2015. The main difference between the conference version and this paper is the addition of pictures.

	\newcommand{\sem}[1]{[\![#1]\!]}
			\newcommand{\seminf}[1]{\sem{#1}_{\omega}}
			\newcommand{\finrange}[1]{\langle \delta_1,\ldots,\delta_{#1} \rangle}
			\newcommand{\infrange}{\langle \delta_1,\delta_2,\ldots \rangle}

	\section{Reduction to limitedness of cost automata}
	\label{sec:reduction}
	This section introduces cost automata, and describes a reduction of the star height problem to a problem for cost automata, called limitedness. The exact reduction used in this paper was shown in~\cite{Kirsten05} and is included to make the paper self-contained. The idea to use counters and limitedness in the study of star height dates back to the original proof of Hashiguchi.
	
	\subsection*{Cost automata and the limitedness problem} We begin by introducing the automaton model used in the limitedness problem. This is  an automaton with counters, which in the literature  appears under the following names: 
	 \emph{nested distance desert automaton} in~\cite{Kirsten05};  \emph{hierarchical B-automaton} in~\cite{BojanczykColcombet06} and~\cite{Colcombet09}; and  \emph{R-automaton} in~\cite{Aky}. (The last model is slightly more general, being equivalent to the non-hierarchical version of B-automata, but its limitedness problem is no more difficult.) Here, we simply  use the name \emph{cost automaton}.	 	Define  a {cost automaton} to be a nondeterministic finite automaton, additionally equipped with a finite  set of \emph{counters} which store natural numbers. For each counter $c$, there is a distinguished subset of transitions that \emph{increment counter $c$}, and a distinguished subset of transitions that \emph{reset counter $c$}. The counters are totally ordered, say they are called $0,\ldots,n$, and every transition that  resets or increments a counter $c$ must also reset all counters $0,\ldots,c-1$. At the beginning of a run, all counters have value zero.

Define the \emph{value} of a run  to be the maximal value which could be seen in some counter during the run. More formally, this is the smallest $m$ such that for every counter $c$, there can be at most $m$ transitions in the run that increment counter $c$ and which are not separated by a transition which resets counter $c$.

\begin{ourexample}\label{ex:max-block}
Here is a picture of a (deterministic) cost automaton with one counter and input alphabet $\set{a,b}$. In the pictures, initial and final states are represented using dangling arrows, and the counter actions are represented in red letters.
	 \mypic{1}
The automaton increments the counter whenever it sees $a$, and resets it whenever it sees $b$. Therefore, for every input word, the value of the unique run of this automaton is the biggest number of consecutive $a$'s that are not separated by any $b$.
\end{ourexample}

\begin{ourexample}\label{ex:two-counters}
Here is a picture of a (deterministic) cost automaton with two  counter called $0,1$ and input alphabet $\set{a,b}$. 
	 \mypic{3}
Therefore, for every input word, the value of the unique run of this automaton is the maximum of the following values:
\begin{itemize}
	\item  biggest number of $a$'s that are not separated by any $b$;
	\item  the number of $b$'s.
\end{itemize}
In particular, if the input word length $n$, then the value of the unique run is at least $\sqrt n$.
\end{ourexample}

\begin{ourexample}\label{ex:min-block}
Here is a picture of a (nondeterministic) cost automaton with one counter and input alphabet $\set{a,b}$. Nondeterminism is used in selecting the initial state, and the moment when the first horizontal transition is used. The only transition that affects the counter is the self-loop  in the middle.
	 \mypic{2}
	 The value of a run for this automaton is the number of $a$'s the was seen after first entering the middle state, until some $b$ is reached. For an input word of the form 
	 \begin{align*}
  a^{n_1} b a^{n_2}  \cdots b a^{n_k}
\end{align*}
the accepting runs have values $n_1,n_2,\ldots,n_k$ respectively. 
\end{ourexample}

	A cost automaton is said to be \emph{limited} over a language $L$ if there exists a bound $m$ such that for every word in  $L$ there exists a run of the cost automaton which is both accepting (i.e.~begins in an initial state and ends in a final state) and has  value at most $m$. 	
 The \emph{limitedness problem} is to decide, given a cost automaton and a regular language $L$, whether the automaton is limited over $L$.  From the perspective of the limitedness problem, it is natural to associate (cf.~\cite{Colcombet09}) to a cost automaton the function $\sem \Aa$  defined by
 \begin{align*}
    [\Aa](w) = \begin{cases}
    	\infty & \mbox{if there is no run accepting run of $\Aa$ over $w$}\\
    	n & \mbox{if $n$ is the minimal value of an accepting run of $\Aa$ over $w$}
    \end{cases}.
\end{align*}
In terms of the function $\sem \Aa$, the  limitedness problem asks if $\sem \Aa$ is bounded over the subset of arguments $L$ (having at least one value $\infty$ means that a function  is unbounded).

\begin{ourexample}
Define a \emph{$b$-block} in a word to be a maximal set of connected positions which all have label $b$. 
If $\Aa$ is the automaton in Example~\ref{ex:max-block}, then $\sem \Aa$ maps a word to the maximal size of an $b$-block; for Example~\ref{ex:min-block} the corresponding function returns the minimal size of a  $b$-block. the  of consecutive $a$'s.  If $\Aa$ is the automaton in Example~\ref{ex:two-counters}, then
\begin{align*}
  \sqrt {|w|} \le \sem \Aa (w) \le |w|.
\end{align*}
In particular, the answer to the limitedness problem for the automaton from Example~\ref{ex:two-counters} is ``limited'' if and only if the set of arguments $L$ is finite.
\end{ourexample}

 The rest of Section~\ref{sec:reduction} is devoted to proving the   reduction stated in the following theorem. This is the same reduction that was used in Kirsten's proof from~\cite{Kirsten05}, a detailed proof can be  found in Proposition 6.8 of~\cite{Kirsten06}.
 
 \begin{thm}\label{thm:reduction}
 	The star height problem reduces to the limitedness problem.
 \end{thm}

 In the reduction we consider  a normal form of regular expressions, which are called \emph{string expressions} following~\cite{Cohen}. The main idea is that in a string expression, a concatenation of unions is converted into a union of concatenations. String expressions have two parameters, called    
(star) height  and degree, and  are defined by induction on   height. A string expression of height  $0$ and degree  at most $m$ is defined to be any regular expression describing a finite set of words of length at most $m$.
	For $h \ge 1$, a string expression of height  at most $h$ and degree at most $m$ is defined to be any finite union of expressions of the form
	\mypic{4}
	
	 A straightforward on $h$, see Lemma 3.1 in~\cite{Cohen}, shows that every regular expression of star height $h$ is equivalent to  a string expression of height $h$ and some degree.

	\newcommand{\univ}[3]{e_{#1}^{#2,#3}}
	\newcommand{\univaut}[2]{#1^{#2}}
	\newcommand{\univautparam}[3]{\Aa_{#1,#2}^{#3}}
	\newcommand{\wdlang}[3]{\Ll_{#1}^{#2,#3}}

	\newcommand{\stringexp}[3]{\invim{#1}_{#2}^{#3}}
	\newcommand{\invim}[1]{[#1]}
 The main observation in the reduction of the star height problem to the limitedness problem is the following proposition, which says that for every star height $h$ and every regular language $L$,   there is a  cost automaton which maps an input word to the smallest degree of a string expression that contains the input word,  has  star height at most $h$,   and is contained in~$L$.
 
	 \begin{prop}\label{lem:univaut} For every regular language $L \subseteq A^*$  and $h \in \Nat$ one can compute a cost automaton $\Aa$ such that $\sem{\Aa}$ maps  a word $w \in A^*$ to the number
	 \begin{align*}
\min\set{\text{degree of $e$} :   \text{$e$ is  a height $\le h$ string expression $e$ satisfying $w \in e \subseteq L$}}.
\end{align*}
	 \end{prop}

Note that the function $\sem \Aa$ always returns a number $\le |w|$ for word $w \in L$, because the word $w$ itself is a string expression of height $0$ and degree $|w|$.

\newcommand{\lseq}{\mathbf L}
\newcommand{\sseq}{\mathbf s}
\newcommand{\kseq}{\mathbf K}
\newcommand{\pset}{\mathsf P}
\newcommand{\comph}[2]{f_{#1}^{#2}}
\begin{proof}
	The proposition is proved by induction on $h$. The induction base of $h=0$ is straightforward: the cost automaton has one counter, has an accepting run if and  only if the word belongs to $L$, and the value of the counter in this run is the length of the input word.
	
	Let us now do the induction step. 	Suppose that we have proved the proposition for $h-1$, and we want to prove it for $h$. Take $L$ as in the statement of the proposition.
Since $L$ is regular, it is recognised by some homomorphism 
\begin{align*}
\alpha : A^* \to M	
\end{align*}
into a finite monoid. For a subset $N \subseteq M$ define
\begin{align*}
\stringexp N h m  \subseteq A^*
\end{align*}
the union of all  string expressions that have degree at most $m$, height at most $h$, and which are contained in $\alpha^{-1}(N)$. (Actually, this union can be viewed as a single string expression, since string expressions of given height and degree are closed under union.) 

In the lemma below,  we use the following notation: for $w_1,\ldots,w_i \in A^*$ and $N_1,\ldots,N_i \subseteq M$, we write
\begin{align*}
	w_1 N_1^* \cdots w_i N_i^* \subseteq M
\end{align*}
to be all possible elements of the monoid  $M$ that can be obtained by taking a  product $\alpha(w_1) n_1 \cdots \alpha(w_i) n_i$ where each $n_j$ is in the submonoid of $M$ generated by $N_j$. The following lemma follows immediately from the definition of string expressions and the set $[N]^m_h$.
\begin{lem} Let $h,m \in \Nat$, $w \in A^*$  and  $N \subseteq M$. Then $w \in [N]^m_h$ if and only if
	\begin{itemize}
		\item[(*)] 	  there exists a number $i \le m$ and  subsets
	\begin{align*}
		N_1,\ldots,N_i \subseteq M
	\end{align*}
	such that $w$ admits a factorisation
	\begin{align*}
		w= w_1 v_1 w_2 v_2 \cdots w_i v_i 
	\end{align*}
	such that the following conditions hold
	\begin{eqnarray}
				w_1 N_1^*  \cdots w_i N_i^* \subseteq \alpha(L) && \label{eq:monoid-product}\\
		|w_j|  \le  m  && \mbox{for  $j \in \set{1,\ldots,i}$} \label{eq:short-factors}\\
		v_j  \in (\stringexp {N_j} {h-1} m)^*  && \mbox{for $j \in \set{1,\ldots,i}$}\label{eq:induction}.
\end{eqnarray}	

	\end{itemize}
\end{lem}
To complete the proof of the proposition,  we will show that for every $N$ and $h$ there exists  a cost automaton, which maps a word $w$ to the smallest  value $m\in \Nat$ which satisfies   condition (*) in the above lemma.
When reading an input word $w$, the automaton proceeds as follows. It uses nondeterminism to guess a factorisation
\begin{align*}
	w = w_1  v_1 w_2 v_2 \cdots w_i v_i ,
\end{align*}
as in the statement of the claim, and it also uses nondeterminism to guess the subsets $N_1,\ldots,N_i$. The subset $N_j$ is only stored in the memory of the automaton while processing the block $v_j$.  Condition~\eqref{eq:monoid-product} is verified using the finite control of the automaton: when the automaton is processing the block $w_j v_j$, its finite memory stores  $N_j$ as well as 
\begin{align*}
	w_1 N_1^*  \cdots w_j N_j^* \subseteq M.
\end{align*}
The remaining conditions are verified using counters. The most important counter is used to the length $i$ of the factorisation, i.e.~it is incremented whenever a block of the form $w_j v_j$ is processed. A less important counter is used to bound the length of the words $w_1,\ldots,w_i$, as required by condition~\eqref{eq:short-factors}, i.e.~it is incremented whenever reading a single letter of $w_i$, and reset after finishing $w_i$.   Finally, for condition~\eqref{eq:induction}, when processing $w_j$ we use the inductively obtained  cost automaton $\Aa_{N_j}$, which works for the smaller height $h-1$. This inductively obtained cost automaton is run in a loop,  corresponding to the Kleene star  in condition~\eqref{eq:induction}, and after each iteration of this loop all of its counters are reset.
	 \end{proof}

Theorem~\ref{thm:reduction}, which reduces the star height problem to the limitedness problem, is a straightforward corollary of the above proposition. 
		 Suppose that we want to decide if a regular language $L$ has star height $h$. 
		 Apply Proposition~\ref{lem:univaut} to $L$ and $h$, yielding a cost automaton. We claim that  $L$ has star height $h$ if and only if the cost automaton is limited over $L$. 
		 
		 If the cost automaton is limited over $L$, say  every word in $L$ admits a run with value at most $m$, then by the proposition, we know that every word $w \in L$ belongs to some string expression $e$ which has height at most $h$ and degree at most $m$. There are finitely many string expressions of given height and degree, and therefore one can take the union of all of these string expressions, yielding a regular  expression (in fact, a string expression) of star height at most $h$  which defines $L$.
		 
		 If $L$ has star height $h$, then it is defined by a regular expression of star height $h$, and therefore also by a string expression of the same height (and some degree), and the cost automaton is limited by the degree.
		 
		 This completes the reduction of the star height problem to the limitedness problem.

	\section{The limitedness problem is decidable}
	\label{sec:limitedness}
	The reduction in the previous section was used already in the proof of Kirsten, and is included here only to make this paper self-contained. This section contains the original contribution of the paper, namely an alternative decidability argument for  the limitedness problem. We show that the limitedness problem  reduces to solving Gale-Stewart games with $\omega$-regular winning conditions,  a problem which was shown decidable by B\"uchi and Landweber~\cite{BuchiLandweber69}. Also in this section, most of the material is not new: apart from using the  the B\"uchi-Landweber theorem as a target of the reduction, the reduction itself is mainly based on a technique from~\cite{ColcombetLoding08}.
	\subsection*{Gale-Stewart Games}  A \emph{Gale-Stewart game} is a game of perfect information and infinite duration played by two players, call them A and B. Each player has their own finite alphabet, call these alphabets $A$ and $B$. The game is played in rounds; in each round player A picks a letter  from alphabet $A$ and player B responds with a letter from alpabet $B$. After $\omega$ rounds have been played, the game results in  an $\omega$-word of the form $a_1b_1 a_2 b_2 \cdots$ with $a_1,a_2, \ldots \in A$ and $b_1,b_2,\ldots \in B$.    The objective of player B is to ensure that the resulting sequence belongs to a given set of $\omega$-words, called the \emph{winning condition for player B}, while the objective of player A is to avoid this.
	
A special case of Gale-Stewart games that is important in computer science  is when the winning condition is $\omega$-regular, i.e.~it is recognised by a nondeterministic B\"uchi automaton, or any other equivalent model.
Let us briefly recall nondeterministic B\"uchi automata, following~\cite[Definition 5.1]{Thomas}. The syntax of such an automaton is exactly the same as for a nondeterministic automaton on finite words. An $\omega$-word over the input alphabet is said to be accepted by the automaton if there is a run (i.e.~a labelling of word positions with transitions which is consistent with the transition relation in the usual sense) which begins in some initial state and uses some accepting state infinitely often. Languages recognised by nondeterministic B\"uchi automata are called $\omega$-regular; this is a robust class of languages which in particular is closed under Boolean operations.
A problem known as the  \emph{Church synthesis problem} is to decide, given a Gale-Stewart game with an $\omega$-regular winning condition, which of the players, if any, has a winning strategy. This problem was solved by B\"uchi and Landweber, who proved  the following theorem.  

\begin{thm}\label{thm:} \cite{BuchiLandweber69}
	Consider a Gale-Stewart game with an $\omega$-regular winning condition.
	\begin{itemize}
		\item The game is determined, i.e.~one of the players has a winning strategy.
		\item One can decide which player has a winning strategy.
		\item The player with a winning strategy has a finite memory winning strategy (see proof of Lemma~\ref{prop:run-system-characterisation} for a definition).
	\end{itemize}
\end{thm}

This is a non-trivial theorem and its proof is not included here. A modern proof, see~\cite[Theorem 6.16]{Thomas}, consists of  two results: memoryless determinacy of parity games, and conversion of nondeterministic B\"uchi automata   into deterministic parity automata. 
Using the B\"uchi-Landweber theorem, we show below that  the limitedness problem is decidable, because it reduces to determining the winner in a Gale-Stewart game with an $\omega$-regular winning condition.
\begin{thm}\label{thm:main}The limitedness problem reduces to solving Gale-Stewart games with $\omega$-regular winning conditions.
\end{thm}

Fix a cost automaton $\Aa$ and a regular language $L$, both  over the same input alphabet. 
Let us begin by fixing some notation. We view a run of the cost automaton, not necessarily accepting, as a sequence of transitions  such that the source state of the first transition is the initial state, and which is consistent in the usual sense: for every transition in the run, its target state is the same as the source state of the following transition. We will also talk about infinite runs.

In the proof, we will be often be using sets of transitions. It is convenient to visualise a set of transitions as a bipartite graph according to the following picture:
\mypic{5}
Note that in the above picture, all transitions in the set use the same input letter; this will always be the case below.
 If $\delta_1,\ldots,\delta_n$ is a sequence of  sets of transitions in the cost automaton, then a run in $\delta_1 \cdots \delta_n$ is defined to be any run where the $i$-th transition is from the set $\delta_i$; likewise for infinite sequences of sets of transitions. Here is a picture:
 \mypic{6} 
 
%

 For a number $m \in \Nat \cup \set \infty$, consider the following Gale-Stewart game, call it the \emph{limitedness game with bound $m$}.
The alphabet for player A is the input alphabet of the cost automaton. The alphabet for player B is 
\begin{align*}
  \bigcup_{a \in \text{input alphabet}} \set{ \delta : \text{$\delta$ is a set of transitions that read letter $a$}}.
\end{align*}
The idea is that player A constructs an infinite word, and player B constructs a representation of a set  of runs that is sufficient to accept  prefixes of this word with  counter values that are small with respect to $m$. This idea is formalised by the winning condition, defined as follows. When $m$ is a finite number, then the winning condition for player B is the set of sequences
\begin{align*}
	a_1 \delta_1 a_2 \delta_2 \cdots
\end{align*}
such that $a_1,a_2,\ldots$ are letters and $\delta_1,\delta_2,\ldots$ are sets of transitions, and the following conditions  hold for every $i$:
	\begin{enumerate}
		\item  every transition in the set $\delta_i$ is over letter $a_i$; and
	\item every run in $\delta_1 \cdots \delta_i$ has value at most $m$; and
	 	\item if $a_1 \cdots a_i \in L$ then  some run in $\delta_1 \cdots \delta_i$ is accepting.
	\end{enumerate}
In particular, the winning condition implies that if $a_1 \cdots a_i$ belongs to $L$, then it admits an accepting run of the cost automaton which has value at most $m$.
For $m = \infty$,  item 2)  is replaced by 
\begin{enumerate}
	\item[2')] for every infinite run in $\delta_1 \delta_2 \cdots$  and every counter $c$, if the run increments $c$  infinitely often, then it resets $c$ infinitely often.
\end{enumerate}

The following lemma is also true for finite values of $m$, but we only use the case of $m=\infty$ in the reduction.
 \begin{lem}\label{lem:}
	For $m =\infty$ the winning condition is $\omega$-regular.
 \end{lem}
 \begin{proof} The class of $\omega$-regular languages is closed under intersections and complements, and therefore it suffices to give separate nondeterministic B\"uchi automata for items 1), 2') and the complement of item 3). For items 1) and 3) one only needs a deterministic automaton with an acceptance condition, sometimes called the \emph{trivial acceptance condition}, where a run is accepting unless as long as it never  visits some designated sink state. (This can be easily converted into a B\"uchi automaton, by making all states accepting except for the sink state.) For item 1), only one additinal state except the sink state is needed; for item 3) it suffices to compute the set of states of the cost automaton that are accessible via runs in $\delta_1 \cdots \delta_i$. It remains to deal with the complement of 2').  The complement of item 2') is recognised by a nondeterministic B\"uchi automaton whose number of states is polynomial in the number of counters; the B\"uchi automaton nondeterministically  guesses a run of the cost automaton which increments a counter infinitely often without resetting it. 
 \end{proof}
 
In particular, thanks to the B\"uchi-Landweber Theorem, one can decide the winner in the limitedness game with bound~$\infty$. 
Therefore, Theorem~\ref{thm:main} will follow from Proposition~\ref{prop:run-system-characterisation} below, which reduces the limitedness problem to finding the winner in the limitedness  game with bound $\infty$. 
		
		\begin{prop}\label{prop:run-system-characterisation}
The following  are equivalent:
			\begin{enumerate}
				\item[(a)] The cost automaton $\Aa$ is limited over language $L$;
				\item[(b)] Player B has a winning strategy in the limitedness game with some finite bound  $m< \infty$.
								\item[(c)] Player B has a winning strategy in  the limitedness game with bound $\infty$.
			\end{enumerate}
		\end{prop}

		A closer analysis shows that condition (c)  can be decided in  exponential time, because the game has a  winning condition whose complement is recognised by a polynomial size B\"uchi automaton, and therefore (using determinisation) the game corresponds to an exponential  size parity game with a polynomial number of ranks~\cite{Thomas}. 
Therefore, our reduction gives an exponential time solution to the limitedness problem, which is inferior to Kirsten's optimal polynomial space solution~\cite{Kirsten05}.
		
		The rest of this paper is devoted to showing Proposition~\ref{prop:run-system-characterisation}.
			We prove equivalence of (b) and (c), and  equivalence of (a) and (b). 
			
				\subsection*{Equivalence of (b) and (c)} The implication  (b) $\Rightarrow$ (c) is immediate. Let us do    (c) $\Rightarrow$ (b).  Since the  limitedness game with bound $\infty$ is a Gale-Stewart  game with an $\omega$-regular winning condition, the B\"uchi-Landweber theorem implies that if player B has a winning strategy, then player B has a finite memory winning strategy. Let us recall the definition of a finite memory winning strategy for  player B. This is  a deterministic finite automaton $\Bb$ whose input alphabet is the alphabet of player A (in our application: the  input  alphabet of the cost automaton) together with a function  $f$ from states of $\Bb$ to the alphabet of player B (in our application: sets of transitions of the cost automaton which share a common input letter)
				such that  when player A chooses letters $a_1, \ldots, a_n$ in the first $n$ rounds, then the  letter chosen by player B in the $n$-th round  is obtained by applying $f$  to the state of $\Bb$ after reading the word $a_1 \cdots a_n$.  

			To prove condition (b) in Proposition~\ref{prop:run-system-characterisation},  consider  the  strategy defined by a deterministic  automaton $\Bb$ which is winning in the limitedness game with bound $\infty$. We will show that when  player B uses this same strategy, the resulting runs have value at most $m$, where  $m$ is the number of states in the cost automaton  times the number of states in the deterministic automaton $\Bb$ defining the strategy.   Toward a contradiction, suppose that player A chooses letters $a_1, a_2,  \ldots$,  player B uses the finite memory strategy given by $\Bb$ to  respond with sets of transitions $\delta_1,\delta_2,\ldots$, and the winning condition for bound $m$ is violated. 
			 This means that for some $n \ge m$,  some finite run $\rho$ in  $\delta_1 \cdots \delta_n$  
			 has  value at least $m$. By the pigeon hole principle, there  must exist $i,j$ with  $i < j$  such that:  in the run of $\Bb$ on $a_1 \cdots a_n$ the same state is used in positions $i$ and $j$; in the run $\rho$ the same state is used in positions $i$ and $j$; and one of the counters is incremented but not reset  between positions $i$ and $j$. Here is a picture:
			 \mypic{7} 
			 Therefore, if player A chooses the infinite word $a_1 \cdots a_{i} (a_{i+1} \cdots a_j)^\omega$, then the response of player B
			 will contain a run where one of the counters is incremented infinitely often but reset finitely often, a contradiction.

						\newcommand{\score}{\mathrm{score}}
			
			\subsection*{Equivalence of (a) and (b)} The implication  (b) $\Rightarrow$ (a) is immediate. To prove  (a) $\Rightarrow$ (b), we use  ideas from~\cite{ColcombetLoding08}. Assume that the automaton is limited over $L$,  i.e.~there is some $m \in \Nat$ such that every word in $L$ admits an accepting run with value at most $m$. For a natural number $m' \ge m$,  define a \emph{scoring function with bound $m'$}  to be a function which maps every finite run to a number in $\Nat \cup \set \infty$, called its \emph{score},  such that runs with value $\le m$ have finite score and  runs with value $>m'$ have infinite score. No constraints are imposed on runs with values between $m$ and $m'$.
A scoring function is called \emph{monotone}  if 
							 \begin{align*}
				\score(\rho) \le \score(\rho') \quad \Rightarrow \quad \score(\rho \sigma) \le \score(\rho' \sigma )
			\end{align*}
			holds for every runs $\rho, \rho'$ that  end in the  same state, and every run $\sigma$  that begins in this state.

			To prove condition (b),  Lemma~\ref{lem:what-loding-needs} shows that the existence of a monotone scoring function is a sufficient condition for player B winning the limitedness game with a finite   bound;    Lemma~\ref{lem:what-loding-has} shows that this sufficient condition can be met.  
\begin{lem}\label{lem:what-loding-needs}
	If there exists a monotone scoring function with bound $m'$, then 
	player B has a winning strategy in the limitedness game with  bound $m'$.

\end{lem}
\begin{proof} (This proof corresponds to Lemma 9 in~\cite{ColcombetLoding08}.)		
Given a scoring function, define a finite run  to be \emph{optimal} if its score  is finite  and  minimal among the scores of  runs that read the same word and end in the same state.
	 Consider the following strategy for player B: if player A has played a  word $w$, then player B responds with the set of transitions $t$ such that some  optimal run over $w$ uses $t$ as the  last transition. (This strategy depends on the choice of scoring function.)  We  show  that this strategy is winning in the limitedness game with bound $m'$.
	 Suppose that  in a finite prefix of a play where player B uses this strategy,  player A has chosen letters $a_1,\ldots,a_i$ and player B has replied with sets of transitions $\delta_1,\ldots,\delta_i$. 
	 
	 \begin{claim}
	 	The set $\delta_1 \cdots \delta_i$  is equal to the set optimal runs over $a_1 \cdots a_i$.
	 \end{claim}
\begin{proof}
	 If a run $\rho$ is optimal, then by monotonicity every prefix of $\rho$ is optimal, and therefore every  transition used by $\rho$ is the last transition of some optimal run,  thus proving membership $\rho \in \delta_1 \cdots \delta_i$. The converse inclusion is proved by induction on $i$. Suppose that  $\delta_1 \cdots \delta_i$ contains a run of the form $\rho t$, where $\rho$ has length $i-1$ and  $t$ is the  last transition used by the run. By induction assumption, $\rho$ is optimal. By definition of $\delta_i$ there is some optimal run of the form $\rho't$. By optimality, $\rho$ has smaller or equal score than  $\rho'$. By monotonicity  $\rho t$ has smaller or equal score than the optimal run $\rho'  t$, and therefore is itself optimal.
\end{proof}

Using the claim, we verify that  the above described strategy of player B is winning in the limitedness game with bound $m'$. Item 1) in the winning condition is clearly satisfied. Since optimal runs have finite score, and runs with finite score have value at most $m'$, item 2) is satisfied for bound $m'$. Since every word $a_1 \cdots a_i \in L$ admits an accepting run with value at most $m$, and runs of value $m$ have finite score, then $a_1 \cdots a_i$ admits an accepting optimal run, and therefore item  3) is satisfied.
\end{proof}

To complete the proof of Proposition~\ref{prop:run-system-characterisation}, we  show that there always exists a monotone scoring function. 

\begin{lem}\label{lem:what-loding-has}
	There exists a monotone scoring function with some finite bound.
\end{lem}
\begin{proof}(This proof corresponds to Lemma 8 in~\cite{ColcombetLoding08}.)
	When   there is a single counter, a good scoring function is one which maps runs of value exceeding $m$ to $\infty$, and maps other runs to the number of increments after the last reset; in this special case $m=m'$. 
	In particular, the material given so far is sufficient to decide limitedness of distance automata, which are the special case of cost automata with one counter and no resets.
	
A more fancy  scoring function is needed when there are two or more  counters.  Suppose that the counters are numbered  $0,\ldots,n$. A scoring function needs to return one number (possibly $\infty$), which somehow summarises all counter values at the same time. One natural idea would be to define the score of a run to be $\infty$ if its value exceeds $m$, and otherwise to be 
	\begin{align*}
		\sum_{i=0}^n a_i \cdot (m+1)^i
	\end{align*}
	assuming that $a_i$ is the number of increments on  counter $i$ after its last reset. In other words, a counter valuation (a tuple in $\set{0,\ldots,m}^n$) would be interpreted as a number in base $m+1$. This scoring  function is  unfortunately
 not monotone.  Indeed, suppose that there are two counters, and consider two runs:
	\begin{itemize}
		\item  $\rho$ does one increment on counter $1$ (score $m+1$);
				\item  $\rho'$ does  $m$ increments on counter $0$ (score $m$).		
	\end{itemize}
	Consider a transition $t$ that increments counter $0$.
	Appending $t$ to these runs reverses the order on scores, because $\rho t$ has score $m+2$, while $\rho' t$ has score $\infty$.
	
Monotonicity is recovered using  the following modification. When an increment to counter $i$  causes its value to reach $m+1$, then instead of setting the whole score to $\infty$, the value of counter $i$ is carried over to counter $i+1$, which may trigger further carries. The score becomes $\infty$ only when the most important counter $n$  reaches $m+1$. 
	
	Formally, the score of a run $\rho$ is defined as follows by induction on length. The score of an empty run is $0$. Suppose that  the score of $\rho$ is already defined and $t$ is a single transition. If the score of $\rho$ is $\infty$, then also the score of $\rho t$ is $\infty$. Otherwise, let
	\begin{align*}
		\sum_{i=0}^n a_i \cdot (m+1)^i
	\end{align*}
 be the  base $m+1$ representation of the score of $\rho$.
If the transition $t$ does not do any counter operations, then the score of $\rho t$ is the same as the score of $\rho$. If the transition $t$ resets some counter  $k \in \set{0,\ldots,n}$, then the score of $\rho t$ is obtained by resetting digits on positions $0,\ldots,k$, i.e.~the score of $\rho t$ is
	\begin{align*}
		\sum_{i=k+1}^n a_i \cdot (m+1)^i
	\end{align*}
    Finally, if  the transition $t$ increments some counter $k \in \set{0,\ldots,n}$, then the score of $\rho t$ is obtained by resetting digits on positions $0,\ldots,k-1$ and incrementing the digit on position $k$, i.e.~the score of $\rho t$ is 
	\begin{align*}
		(a_{k}+1) \cdot (m+1)^{k} + 		\sum_{i=k+1}^n a_i \cdot (m+1)^i
	\end{align*}
	If, as a result of the increment, the score reaches the maximal number, call it $m'$, which can be stored in base $m+1$ representation using digits $a_0,\ldots,a_n$, then  the score of $\rho t$ is defined to be $\infty$.  We show below  that the function thus defined is  a monotone scoring function with the bound being $m'$ defined above. 
	\begin{itemize}
		\item Monotonicity.  We need to show 
		\begin{align*}
				\score(\rho) \le \score(\rho') \qquad \Rightarrow \qquad \score(\rho t) \le \score(\rho't ).
		\end{align*}
		The most interesting case is when the transition $t$ increments some counter $k$.  In this case, the score is obtained by first  resetting digits $0,\ldots,k-1$, which is monotone, and then adding $(m+1)^k$, which is also monotone.
		\item Runs with value at most $m$ have finite score. By induction on the length of the run, one shows that if a run has value at most  $m$, then for every $i \in \set{0,\ldots,n}$, the $i$-th digit in the base $m+1$ representation of its score is equal to the number of increments on counter $i$ after its last reset.
		\item Runs with finite score have value at most $m'$.  Toward a contradiction, suppose that a run with finite score has value strictly bigger than $m'$. This means that for some counter $i$, the run does more than $m'$ increments without any reset of counter $i$.  During these increments,  the score will carry over so that the value of counter $n$ exceeds $m$, and therefore the score will become infinite.
	\end{itemize}
\end{proof}

This completes the proof of equivalence of (a) and (b) in Proposition~\ref{prop:run-system-characterisation}, and therefore also the proof of Theorem~\ref{thm:reduction}.

%
%
%
%
%
%
%
%

\newcommand{\noopsort}[1]{}


\begin{thebibliography}{Has88}

\bibitem[AKY08]{Aky}
Parosh Aziz Abdulla, Pavel Krcal and Wang Yi.  
\newblock R-Automata.
\newblock In {\em CONCUR}, pages 67--81, 2008.

\bibitem[BL69]{BuchiLandweber69}
J. Richard B\"uchi and Lawrence H. Landweber.
\newblock Solving Sequential Conditions by Finite-State Strategies.
\newblock {\em Transactions of the American Mathematical Society}, Volume 138, pages 295--311, 1969.

\bibitem[BC06]{BojanczykColcombet06}
Miko\l{}aj Boja\'{n}czyk and Thomas Colcombet.
\newblock Bounds in $\omega$-regularity.
\newblock In {\em LICS}, pages 285--296, 2006.

\bibitem[Coh70]{Cohen}
Rina S. Cohen.
\newblock Star height of certain families of regular events.
\newblock {\em Journal of Computer and System Sciences}, Volume 4, Issue 3, pages 281--297, 1970.

\bibitem[CL08]{ColcombetLoding08}
Thomas Colcombet and Christof L\"oding.
\newblock The nesting depth of disjunctive $\mu$-calculus for tree languages and the limitedness problem.
\newblock In {\em CSL}, pages 416--430, 2008.

\bibitem[Col09]{Colcombet09}
Thomas Colcombet.
\newblock The theory of stabilisation monoids and regular cost functions.
\newblock In {\em ICALP (2)}, pages 139--150, 2009.

\bibitem[Has88]{DBLP:journals/iandc/Hashiguchi88}
Kosaburo Hashiguchi.
\newblock Algorithms for determining relative star height and star height.
\newblock {\em Inf. Comput.}, 78(2):124--169, 1988.

\bibitem[Kir05]{Kirsten05}
Daniel Kirsten.
\newblock Distance desert automata and the star height problem.
\newblock {\em Informatique Th\'eorique et Applications}, 39(3):455--509, 2005.

\bibitem[Kir06]{Kirsten06}
Daniel Kirsten.
\newblock Distance Desert Automata and Star Height Substitutions.
\newblock Habilitation thesis, Universit\"at Leipzig.

\bibitem[Sim88a]{Simon88}
Imre Simon. 
\newblock Recognizable sets with multiplicities in the tropical semiring. 
\newblock In {\em Mathematical Foundations of Computer Science}, pages 107–-120, 1988.

\bibitem[Sim88b]{Simon88b}	
Imre Simon.
\newblock Properties of factorization forests. 
\newblock {\em Formal Properties of Finite Automata and Applications}, pages 65--72, 1988.

\bibitem[Tho96]{Thomas}
Wolfgang Thomas.
\newblock Languages, automata and logic.
\newblock Technical Report 9607, Institut f\"ur Informatik und Praktische Mathematik, Christian-Albrechts-Universit\"at zu Kiel,  1996.

\bibitem[Tor11]{Torunczyk11}
Szymon Toru\'nczyk.
\newblock Languages of profinite words and the limitedness problem.
\newblock {\em PhD Thesis, University of Warsaw},  2011.


\end{thebibliography}
\end{document}